\documentclass[11pt]{article}
\usepackage[english]{babel}
\usepackage[colorlinks,linkcolor=blue,citecolor=blue,urlcolor=blue]{hyperref}
\usepackage{amsmath,amssymb,graphicx,stmaryrd,cleveref,latexsym,fullpage}
\usepackage{amsopn,amsthm,upref}
\usepackage{microtype}
\usepackage{thmtools}

\usepackage{cite}

\newtheorem{corollary}{Corollary}
\newtheorem{lemma}{Lemma}

\hypersetup{
    colorlinks=true,
    linkcolor=blue,
    urlcolor=blue,
}

\newcommand{\cC}{\mathcal{C}}
\newcommand{\cL}{\mathcal{L}}
\newcommand{\norm}[1]{\left\|#1\right\|}

\newcommand{\argmin}{\mathop{\mathrm{argmin}}}
\newcommand{\cCpow}[1]{\cC^{\otimes #1}}
\newcommand{\PP}{\mathsf{P}}
\newcommand{\EE}{\mathop{\mathsf{E}}}

\title{High-dimensional Expansion of Product Codes is Stronger than Robust and Agreement Testability}
\author{Gleb~Kalachev\thanks{Gleb~Kalachev is with the Faculty of Mechanics and Mathematics, Moscow State University, Moscow, Russia.%
}}
\begin{document}

\newcommand{\F}{\mathbb{F}}
\newcommand{\RR}{\mathbb{R}}
\newcommand{\FF}{\mathbb{F}}
\newcommand{\ZZ}{\mathbb{Z}}
\newcommand{\NN}{\mathbb{N}}
\newcommand{\CC}{\mathbb{C}}
\newcommand{\RX}{R}

\newcommand{\absbr}[1]{\left|#1\right|}
\newcommand{\rbr}[1]{\left(#1\right)}
\newcommand{\sbr}[1]{\left[#1\right]}
\newcommand{\abr}[1]{\left\langle#1\right\rangle}
\newcommand{\floorbr}[1]{\left\lfloor #1\right\rfloor}
\newcommand{\ceilbr}[1]{\left\lceil #1\right\rceil}
\newcommand{\fbr}[1]{\left\{#1\right\}}

\newcommand{\bigabsbr}[1]{\bigl|#1\bigr|}
\newcommand{\bigrbr}[1]{\bigl(#1\bigr)}
\newcommand{\bigsbr}[1]{\bigl[#1\bigr]}
\newcommand{\bigfbr}[1]{\bigl\{#1\bigr\}}
\newcommand{\bigabr}[1]{\bigl\langle#1\bigr\rangle}
\newcommand{\bigfloorbr}[1]{\bigl\lfloor #1\bigr\rfloor}
\newcommand{\bigceilbr}[1]{\bigl\lceil #1\bigr\rceil}

\newcommand{\biggabsbr}[1]{\biggl|#1\biggr|}
\newcommand{\biggrbr}[1]{\biggl(#1\biggr)}
\newcommand{\biggsbr}[1]{\biggl[#1\biggr]}
\newcommand{\biggfbr}[1]{\biggl\{#1\biggr\}}
\newcommand{\biggabr}[1]{\biggl\langle#1\biggr\rangle}
\newcommand{\biggfloorbr}[1]{\biggl\lfloor #1\biggr\rfloor}
\newcommand{\biggceilbr}[1]{\biggl\lceil #1\biggr\rceil}

\newcommand{\Bigabsbr}[1]{\Bigl|#1\Bigr|}
\newcommand{\Bigrbr}[1]{\Bigl(#1\Bigr)}
\newcommand{\Bigsbr}[1]{\Bigl[#1\Bigr]}
\newcommand{\Bigfbr}[1]{\Bigl\{#1\Bigr\}}
\newcommand{\Bigabr}[1]{\Bigl\langle#1\Bigr\rangle}
\newcommand{\Bigfloorbr}[1]{\Bigl\lfloor #1\Bigr\rfloor}
\newcommand{\Bigceilbr}[1]{\Bigl\lceil #1\Bigr\rceil}

\newcommand{\Biggabsbr}[1]{\Biggl|#1\Biggr|}
\newcommand{\Biggrbr}[1]{\Biggl(#1\Biggr)}
\newcommand{\Biggsbr}[1]{\Biggl[#1\Biggr]}
\newcommand{\Biggfbr}[1]{\Biggl\{#1\Biggr\}}
\newcommand{\Biggabr}[1]{\Biggl\langle#1\Biggr\rangle}
\newcommand{\Biggfloorbr}[1]{\Biggl\lfloor #1\Biggr\rfloor}
\newcommand{\Biggceilbr}[1]{\Biggl\lceil #1\Biggr\rceil}

\newcommand{\la}{\leftarrow}
\newcommand{\La}{\Leftarrow}
\newcommand{\ra}{\rightarrow}
\newcommand{\ov}[1]{\overline{#1}}
\newcommand{\Ra}{\Rightarrow}
\newcommand{\wt}{\widetilde}
\newcommand{\half}[1][1]{\frac{#1}{2}}
\newcommand{\Array}[2]{\begin{array}{#1}#2\end{array}}
\newcommand{\alignR}[1]{\begin{flushright}#1\end{flushright}}
\newcommand{\alignC}[1]{\begin{center}#1\end{center}}
\newcommand{\alignL}[1]{\begin{flushleft}#1\end{flushleft}}
\newcommand{\eqsys}[2]{\left\{\Array{l}{#1\\#2}\right.}
\newcommand{\eqsysc}[4]{\left\{\Array{ll}{#1&#2\\#3&#4}\right.}
\newcommand{\eqsysccc}[6]{\left\{\Array{ll}{#1&#2\\#3&#4\\#5&#6}\right.}
\newcommand{\eqsyscccc}[8]{\left\{\Array{ll}{#1&#2\smallskip\\#3&#4\smallskip\\#5&#6\smallskip\\#7&#8}\right.}
\newcommand{\eps}{\varepsilon}
\newcommand{\ds}{\displaystyle}
\newcommand{\sign}{\mathop{\rm sign}\nolimits}
\newcommand{\tr}{\kappa}
\newcommand{\yy}[1][y]{\frac{#1}{\sqrt{1+#1^2}}}
\newcommand{\uplim}{\mathop{\overline{\lim}}\limits}
\newcommand{\downlim}{\mathop{\underline{\lim}}\limits}
\newcommand{\llim}[1]{\lim\limits_{#1\ra\infty}}
\newcommand{\upllim}[1]{\uplim_{#1\ra\infty}}
\newcommand{\downllim}[1]{\downlim_{#1\ra\infty}}
\renewcommand{\phi}{\varphi}
\newcommand{\teta}{\theta}
\newcommand{\comment}[2][red]{\textcolor{#1}{/*#2*/}}
\newcommand{\cc}[1]{#1^{\bot}}

\let\le\leqslant
\let\ge\geqslant
\let\eps\varepsilon
\let\phi\varphi
\let\kappa\varkappa
\let\si\sigma
\let\Om\Omega
\let\geq\geqslant
\let\leq\leqslant
\let\emptyset\varnothing

\newcommand{\CSS}{\mathrm{CSS}}

\newcommand{\rank}{\mathop{\mathrm{rank}}}
\newcommand{\im}{\mathop{\mathrm{im}}}
\newcommand{\underbr}[2]{\vphantom{#1}\smash{\underbrace{#1}_{#2}}}

\theoremstyle{definition}
\newtheorem*{definition}{Definition}

\maketitle
\begin{abstract}
    We study the coboundary expansion property of product codes called product expansion, which played a key role in all recent constructions of good qLDPC codes. It was shown before that this property is equivalent to robust testability and agreement testability for products of two codes with linear distance. First, we show that robust testability for product of many codes with linear distance is equivalent to agreement testability. Second, we provide an example of product of three codes with linear distance which is robustly testable but not product expanding.
\end{abstract}
\section{Introduction}
Recent constructions of asymptotically good locally testable codes (LTC) and quantum LDPC (qLDPC) codes \cite{Panteleev&Kalachev:stoc2022,Dinur:stoc2022,Dinur:decoders,Leverrier:focs2022,Leverrier:qldpcdecoder:2023,Leverrier:2023,Gu:stoc2023:qpdpc-decoder} use a special property of product codes that has several names and definitions: robust testability, agreement testability, and product expansion. It was shown \cite[Lemma 2.9]{Dinur:stoc2022}, \cite[Lemma 1]{PK2022robust} that these definitions are essentially equivalent in the case of the product of two codes. For all known constructions of good qLDPC codes, this property is necessary to get the linear distance and efficient decoders for them. For LTCs there is one exception: in \cite{Lin2022:losslessLTC} the construction of LTC codes is based on one-sided lossless expanders and does not require local codes satisfying specific property.

In \cite[Appendix B]{PK2022robust} it was shown that product expansion can be understood as a form of high-dimensional expansion called coboundary expansion (for 2-dimensional case see also \cite[Section 2.6]{Dinur:decoders}).
Thus, it seems to be an important property of the product code, as well as robust and agreement testability. Moreover, as product expansion is a form of high-dimensional expansion, it is likely to be useful to construct high-dimensional analogs of codes from \cite{Dinur:stoc2022,Panteleev&Kalachev:stoc2022} which could potentially give good quantum locally testable codes (qLTC). 

Also, in \cite[Lemma 1]{PK2022robust} it was shown that product expansion for a pair of codes coincides with agreement testability with the same constant (see also \cite[Section 2.6]{Dinur:decoders}).
The goal of this paper is to clarify the relation between robust testability, agreement testability, and product expansion for the product of more than two codes.
In particular, we consider a natural generalization of agreement testability for product of multiple codes and show that in the case of the product of 3 or more codes: 1) product expansion is different from robust and agreement testability; 2) agreement testability is equivalent to robustness of the axis-parallel line test up to a constant factor.

\subsection{Product expansion}
Here we will give the definition of product expansion from \cite{PK2022robust}. 
The history and relation with other forms of this definition can also be found in \cite{PK2022robust}.
Given linear codes $\cC_1,\dots,\cC_m$ over $\F_q$ we can define the (\emph{tensor}) \emph{product code} 
\[
\cC_1\otimes\dots\otimes\cC_m := \{c\in \F_q^{n_1\times\dots\times n_m} \mid \forall i\in [m]\ \forall \ell\in \cL_i\colon c|_\ell\in  \cC_i\},
\]
where $\F_q^{n_1\times\dots\times n_m}$ is the set of functions $c\colon[n_1]\times \dots\times [n_m] \to \F_q$  and $\cL_i$ is the set of lines parallel to the $i$-th axis in the $m$-dimensional grid $[n_1]\times \dots\times [n_m]$, i.e., 
\[
\cL_i := \{\{x + s\cdot e_i \mid s \in [n_i] \} \mid x\in [n_1]\times \dots\times [n_m], x_i = 0 \}.
\]
Here $e_i$ denotes the vector $(0,\dots,0,1,0\dots,0) \in [n_1]\times \dots \times [n_m]$ with $1$ at the $i$-th position. 

As in \cite{PK2022robust}, for linear codes $\cC_1\subseteq \F_q^{n_1}$, $\cC_2\subseteq \F_q^{n_2}$ we denote by $\cC_1\boxplus \cC_2$ the code $(\cC_1^\bot\otimes \cC_2^\bot)^\bot=\cC_1\otimes \F_q^{n_2}+\F_q^{n_1}\otimes \cC_2 \subseteq \F_q^{n_1\times n_2}$.
Given a~collection $\cC = (\cC_i)_{i\in [m]}$ of linear codes over~$\F_q$, we can define the codes 
\[
\cC^{(i)} := \F_q^{n_1} \otimes\dots\otimes \cC_i \otimes\dots\otimes \F_q^{n_m} = \{ c\in \F_q^{n_1\times\dots\times n_m} \mid \forall \ell\in \cL_i\colon c|_\ell\in  \cC_i \} .
\]
It is clear that $\cC_1\otimes \dots \otimes \cC_m = \cC^{(1)}\cap \dots \cap \cC^{(m)}$ and $\cC_1\boxplus \dots \boxplus \cC_m = \cC^{(1)}+ \dots + \cC^{(m)}$. Note that every code $\cC^{(i)}$ is the direct sum of $|\cL_i| = \frac{1}{n_i}\prod_{i\in [m]} n_i$ copies of the code $\cC_i$. 
For $x\in \F_q^{n_1\times\dots\times n_m}$ we denote by $|x|_i$ and $\norm{x}_i$, respectively, the number and the fraction of the lines $\ell\in\cL_i$ such that $a|_\ell \ne 0$. 
It is clear that $\norm{x}_i = \frac{1}{|\cL_i|}|x|_i$. By $|x|$ and $\norm{x}$ we denote, respectively, the \emph{Hamming weight} (i.e., the number of non-zero entries) and the \emph{normalized Hamming weight} (i.e., the fraction of non-zero entries) of $x$. 
We will also use the following notations: the normalized distance $\delta(x,y):=\norm{x-y}$, the normalized distance to code $\delta(x,\cC):=\min_{y\in\cC}\norm{x-y}$, and the normalized minimum distance $\delta(\cC):=\min_{x\in\cC}\norm{x}$ for a code $\cC\subseteq\F_q^n$.

\begin{definition}[Product-expansion \cite{PK2022robust}]
Given a~collection $\cC = (\cC_i)_{i\in [m]}$ of linear codes $\cC_i\subseteq \F_q^{n_i}$, we say that $\cC$ is \emph{$\rho$-product-expanding} if every codeword $c\in \cC_1\boxplus \dots \boxplus \cC_m$ can be represented as a~sum $c = \sum_{i\in[m]} a_i$, where $a_i\in \cC^{(i)}$ for all $i\in [m]$ and the following inequality holds:
\begin{equation}\label{eq:prod-exp}
\rho\sum_{i\in [m]} \norm{a_i}_i \le \norm{c} .    
\end{equation}
\end{definition}

We denote as $\rho(\cC)$ the maximal $\rho$ such that $\cC$ is $\rho$-product-expanding. In \cite[Appendix B]{PK2022robust} it was shown that $\rho(\cC)$, up to the constant factor $1/m$, is equal to the Cheeger constant of the chain complex naturally associated with the product code $\cC_1\otimes\cdots\otimes \cC_m$.

\subsection{Robust and agreement testability}
Let $X$ be some finite index set, which we will use to enumerate bits of the code. So, a code $\cC\subseteq \F_q^X$ is a set of functions $f:X\to\F_q$. If $I\subseteq X$, then $\cC|_I:=\{c|_I\mid c\in \cC\}$ is punctured code $\cC$ consisting of restrictions of codewords from the code $\cC$ to the index set $I$.

\begin{definition}
    A \emph{test for a code} $\cC\subseteq \F_q^X$ is a set $T\subseteq 2^{X}$ equipped with probability measure $\PP$ on it.
\end{definition}
In this paper, we will always use the following probability distribution:
\begin{equation}\label{eqn:prob}
\PP(I)=\frac{|I|}{\sum_{J\in T}|J|}\quad\mbox{for}\quad I\in T.
\end{equation}

The tester for the pair (code $\cC\subseteq \F_q^X$ and a test $T$) works as follows: for a given word $c\in \F_q^X$ we randomly choose a set $I\in T$ 
and accept $c$ if $c|_I\in \cC|_I$ and reject otherwise. Thus, if $c\in \cC$, then any tester accepts it with probability 1.

\begin{definition}[Test robustness]
    The test $T$ for a code $\cC\subseteq \F_q^X$ is \emph{$\alpha$-robust} if for all $c\in \F_q^X$ we have
    \[\EE_{I\in T}\delta(c|_I,\cC|_I)\ge\alpha \delta(c,\cC),\]
    where $\EE$ denotes expectation.    
\end{definition}
Let us define the maximal robustness:
\begin{align*}
  \rho_r(T,\cC)&:=\max\fbr{\alpha\mid \mbox{Test $T$ is $\alpha$-robust for the code }\cC}.
\end{align*}

Usually, when the code $\cC$ is defined by a set of local codes, the natural test contains supports of all these local codes. For example, product code $\cC_1\otimes \cC_2$ can be defined by local codes on axis-parallel lines of the set $X=[n_1]\times [n_2]$:
$$\cC_1\otimes \cC_2=\fbr{f\in\F_q^{[n_1]\times [n_2]}\Bigm| f(\cdot, j)\in\cC_1\mbox{ for }j\in [n_2],f(i,\cdot)\in\cC_2\mbox{ for }i\in [n_1]}.$$
Thus, the natural test for the code $\cC_1\otimes \cC_2$ is the set of all axis-parallel lines:
$$T=\cL_1\cup\cL_2 = \fbr{[n_1]\times \{j\}\mid j\in [n_2]}\cup \fbr{\{i\}\times [n_2]\mid i\in[n_1]},$$
and $\PP$ defined in \eqref{eqn:prob} corresponds to the following procedure: choose a random direction, then choose a random line along this direction.
This test is called the \emph{axis-parallel line test}. For product of $m\ge 3$ codes, there exist different natural tests, since we can consider axis-parallel subspaces of different dimensions from $1$ to $m-1$.
The following definition gives a straightforward generalization of the 2-flat test from \cite[Algorithm 12.2]{Lin-func-prop}.
\begin{definition}[Axis-parallel $k$-flat test]
    Let $X=[n_1]\times \cdots\times [n_m]$, $k\in [m-1]$. Then, the \emph{axis-parallel $k$-flat test} is defined as the set $T_m^k$ of all $k$-dimensional axis-parallel subspaces (\emph{$k$-flats}) in $X$:
    $$T_m^k(X)=\bigcup_{I\subseteq [m], |I|=k}\cL_I,\quad \cL_I=\Bigfbr{\Bigfbr{x+\sum_{i\in I}s_ie_i\Bigm| s_i\in [n_i]\mbox{ for }i\in I}\Bigm| x\in X, x_i=0\mbox{ for }i\in I}.$$
    We will omit the argument of $T_m^k$ where it is not important or is clear from context.
\end{definition}

Here we follow the terminology from \cite{Lin-func-prop}.
The test $T_2^1$ is the standard axis-parallel line test, $T_m^1$ is its multidimensional version, and $T_m^{m-1}$ is the \emph{axis-parallel hyperplane test}. 
In \cite[Theorem 12.5]{Lin-func-prop} it was shown that $\rho_r(T_m^2, \cC^{\otimes m})\ge \alpha(\delta(\cC),m)$ for $m\ge 3$ and some function\footnote{From the proof of \cite[Theorem 12.5]{Lin-func-prop} it follows that $\alpha(\epsilon, m)=\epsilon^{\frac12(m-2)(m+3)}24^{2-m}$.} $\alpha(\epsilon, m)>0$.
This result shows that the requirement of constant robustness of the test $T_m^2$ for a family of codes $\cC_i^{\otimes m}$ is equivalent to the requirement of linear minimum distance of codes in this family. 
So, the only test that gives a non-trivial requirement on the code $\cC$ is the axis-parallel line test $T_m^1$. 
The test $T_m^1$ can be considered as the composition of tests $T_m^2$ for $\cC^{\otimes m}$ and $T_2^1$ for $\cC^{\otimes 2}$. As it will be shown more formally in Lemma \ref{lemma:robust-combine}, 
\[\rho_r(T_m^1,\cC^{\otimes m})\ge \rho_r(T_m^2,\cC^{\otimes m})\rho_r(T_2^1,\cC^{\otimes 2}),\]
that is, the constant robustness of $T_m^1$ for $\cC^{\otimes m}$ is equivalent to the constant robustness of $T_2^1$ for $\cC^{\otimes 2}$.

The following definition of agreement testability for product of several codes is a straightforward generalization of agreement testability for product of 2 codes \cite[Definition 2.8]{Dinur:stoc2022}.
\begin{definition}[Agreement testability for product code]
    Let $\cC=(\cC_1,\ldots,\cC_m)$ be a collection of codes. Product code $\otimes \cC$ is $\alpha$-\emph{agreement testable} if for each $c_1\in\cC^{(1)},\ldots,c_m\in\cC^{(m)}$ there exists $c\in\otimes\cC$ such that
    \[
    \alpha \EE_{i\in[m]}\|c_i-c\|_i
    \le \EE_{i,j\in[m]}\|c_i-c_j\|,
    \]
    where the uniform distribution on $[m]$ is assumed.
    Let us define the maximal agreement testability:
    \begin{align*}
      \rho_a(\otimes\cC)&:=\max\fbr{\alpha\mid \mbox{product code $\otimes\cC$ is $\alpha$-agreement testable}}.
    \end{align*}
\end{definition}

Note that $\rho_a(\otimes\cC)\le 2$, since $\|c_i-c_j\|\le \|c_i-c\|+\|c_j-c\|\le \|c_i-c\|_i+\|c_j-c\|_j$.

\newcommand{\cI}{\mathcal{I}}
\begin{restatable}[Robust testability + Linear distance = Agreement testability]{lemma}{LemRobustAgreement}\label{lemma:robust=agreement}
    Let $\cC=(\cC_1,\ldots,\cC_m)$ be a collection of codes $\cC_i\in\F_q^{n_i}$, $\rho_r:=\rho_r(T_m^1,\otimes\cC)$, $\rho_a:=\rho_a(\otimes\cC)$. Then
    \[
    \rho_r\ge \frac14\rho_a,\qquad \rho_a\ge \frac{\rho_r}{\rho_r+1}\min_{i\in[m]}\delta(\cC_i).
    \]
\end{restatable}
The proof is given in Appendix \ref{app:robust=agreement}. It is essentially the same as the proof for the product of two codes \cite[Lemma 2.9]{Dinur:stoc2022}.
From Lemma \ref{lemma:robust=agreement} we see that robust and agreement testability are essentially the same. 
Our main result is that product expansion of a collection of codes is different from robust and agreement testability of the product of these codes.
\begin{restatable}{theorem}{MainTheorem}\label{th:main}
    Let $C_t$ be the primitive Reed-Solomon $[n_t,\frac{n_t}{3}]$ code over the field $\F_{2^{2t}}$ defined by the check polynomial $(x-1)(x-\omega)\ldots (x-\omega^{\frac{n_t}{3}-1})$, where $t\in\NN$, $n_t=2^{2t}-1$, $\omega$ is a primitive element of $\F_q$. 
    For each $m\ge 3$ there exist $\alpha_r>0$ and $\alpha_a>0$ such that for all $t\in\NN$ the following inequalities hold:
    \begin{enumerate}
        \item $\rho(\underbrace{C_t,\ldots,C_t}_{m\ \mbox{\scriptsize\rm times}})\le \displaystyle\frac{1}{n_t}$;
        \item $\rho_r(T_m^k, C_t^{\otimes m})\ge \alpha_r$ for all $k\in[m-1]$;
        \item $\rho_a(C_t^{\otimes m})\ge \alpha_a.$
    \end{enumerate}
\end{restatable}
Moreover, product expansion implies robustness of the test $T_m^1$ for $\cCpow{m}$.
\begin{restatable}{proposition}{MainProp}\label{prop:main}
    Let $\cC\subsetneq\F_q^n$ and $m\ge 2$. Then there exists a function $\alpha$ such that $\alpha(x)>0$ for $x>0$ and
    \[
    \rho_r(T_m^1,\cCpow{m})\ge \alpha(\rho(\underbrace{\cC,\ldots,\cC}_{m\ \mbox{\scriptsize\rm times}})).
    \]
\end{restatable}
This proposition together with Theorem \ref{th:main} shows that the product expansion property imposes a stronger constraint on the code $\cC$ than robust testability.
Theorem \ref{th:main} and the Proposition \ref{prop:main} are proved in the next section.

\section{The proofs}
Let us fix $t\in \NN$ and consider the primitive Reed-Solomon $[n,k]$ code $C$ over the field $\F_q$, where $q=2^{2t}$, $n=q-1$, and $k=n/3$. This code can be defined by the check polynomial $p(x)=(x-1)(x-\omega)\ldots (x-\omega^{k-1})$, where $\omega$ is a primitive element of $\F_q$:
\[
C = \Bigfbr{(a_i)_{i=0}^{n-1}\in \F_q^n\Bigm| p(x)\sum_{i=0}^{n-1}a_ix^i\equiv 0\mod (x^n-1)}.
\]
First, we will show that $\rho(C,C,C)\le 1/n$.

First, let us describe the dual of the product of cyclic codes in terms of check polynomials. Consider cyclic codes $\cC_1,\ldots,\cC_m\in \F_q^n$ defined, respectively, by check polynomials $p_1,\ldots,p_m\in \F_q[x]$ such that $p_i|(x^n-1)$:
\begin{align*}
    \cC_i&=\Bigfbr{(a_i)_{i=0}^{n-1}\in \F_q^n\Bigm| p_i(x)\sum_{i=0}^{n-1}a_ix^i\equiv 0\mod (x^n-1)}\\
    &\cong \Bigfbr{a\in \F_q[x]\Bigm| \deg a<n,\ p_i(x)a(x)\equiv 0\mod (x^n-1)}.
\end{align*}
Here for codes $\cC_1\subseteq V_1$, $\cC_2\subseteq V_2$ we say that $\cC_1\cong\cC_2$ if 
there is a linear isomorphism $\phi:V_1\to V_2$ preserving the Hamming distance\footnote{Distinguished bases in $V_1,V_2$ are necessary to define the Hamming distance and the minimum distance of $\cC_1$, $\cC_2$. In the space of polynomials of degree at most $k$ the distinguished basis is $\{1,x,\ldots,x^k\}$.} such that $\phi(\cC_1)=\cC_2$. 

\begin{lemma}\label{lemma:sum-check}
    Let $\cC=\cC_1\boxplus \cdots\boxplus \cC_m$. Consider the ideal $\cI=(x_1^n-1,\ldots,x_m^n-1)\subseteq \F_q[x_1,...,x_m]$. Then
    \begin{equation*}
        \cC = \Bigfbr{a\in \F_q[x_1,...,x_m] \Bigm| \deg_{x_i}a<n \mbox{ \rm and } a(x_1,...,x_m)\prod_{i=1}^m p_i(x_i)\equiv 0 \mod \cI}.
    \end{equation*}
\end{lemma}
\begin{proof}
    For a polynomial $p(x_1,...,x_k)$ define $p^*(x_1,...,x_k):=p(x_1^{n-1},...,x_k^{n-1})\mod \cI$.
    
    Since $p_i(x)$ is a check polynomial for $\cC_i$, then $p_i^*(x)$ is a generator polynomial for $\cC_i^\bot$, i.e.
    \[
    \cC_i^\bot=\bigfbr{p^*_i(x)q(x)\bigm| \deg q<n-\deg p_i}=\bigfbr{a\in \F_q[x]\bigm| \deg a<n\mbox{ and }p_i^*|a}.
    \]
    Hence, the tensor product of $\cC_1^\bot,\ldots,\cC_m^\bot$ is generated by $p^*_1(x_1)\cdots p^*_m(x_m)\in \F_q[x_1,\ldots,x_m]$:
    \[
    \cC_1^\bot\otimes\cdots\otimes \cC_m^\bot=\bigfbr{a\in \F_q[x_1,...,x_m]\bigm| \deg_{x_i} a<n\mbox{ and }p_i^*(x_i)|a}.
    \]
    Therefore, $(p^*_1(x_1)\cdots p^*_m(x_m))^*=p_1(x_1)\cdots p_m(x_m)$ is a check polynomial for $(\cC_1^\bot\otimes\cdots\otimes \cC_m^\bot)^\bot=\cC_1\boxplus \cdots\boxplus \cC_m$.
\end{proof}

\begin{lemma}\label{lemma:not-expanding}
    Let $C$ be the primitive Reed-Solomon $[n,k]$ code $C$ over the field $\F_q$ defined by the check polynomial $p(x)=(x-1)(x-\omega)\ldots (x-\omega^{k-1})$, where $q=2^{2t}$, $n=q-1$, $k=n/3$. Then
    $$\rho(C,C,C) \le 1/n.$$
\end{lemma}
\begin{proof}
    A codeword of the code $C\boxplus C\boxplus C$ can be defined as a polynomial 
    $f(x,y,z)$ such that 
    $$f(x,y,z)p(x)p(y)p(z)\equiv 0\mod (x^n-1,y^n-1,z^n-1).$$ 
    Consider the polynomials 
    $$a'(x,y,z)=\sum_{i=0}^{n-1}\sum_{j=0}^{n-1}\sum_{l=0}^{n-1}a'_{ij\ell}x^iy^jz^l,\quad\mbox{and}\quad a(x,y,z):=a'(x,\omega^{-k}y,\omega^{-2k}z),$$
    where 
    $$a'_{ijl}=\begin{cases}1,&i+j+l\equiv 0\mod n\\0,&\mbox{otherwise}.\end{cases}$$
    First, we will show that $a$ is a codeword of the code $C\boxplus C\boxplus C$. We need to show that 
    \begin{equation}\label{eqn:a-is-codeword}
     a(x,y,z)p(x)p(y)p(z)=0\mod (x^n-1,y^n-1,z^n-1).
    \end{equation}
    Consider the polynomials 
    $$r(x):=p(\omega^k x)=\omega^{k^2}\prod_{i=k}^{2k-1}(x-\omega^i),\qquad s(x):=p(\omega^{2k}x)=\omega^{2k^2}\prod_{i=2k}^{3k-1}(x-\omega^i).$$
    We have $a(x,y,z)p(x)p(y)p(z)=a'(x,\omega^{-k}y,\omega^{-2k}z)p(x)r(\omega^{-k}y)s(\omega^{-2k}z)$, $\omega^n=1$, hence by the replacement $y\mapsto \omega^{-k}y$, $z\mapsto\omega^{-2k}z$ the condition \eqref{eqn:a-is-codeword} can be rewritten as
    \begin{equation}\label{eqn:a'-is-codeword}
     a'(x,y,z)p(x)r(y)s(z)=0\mod (x^n-1,y^n-1,z^n-1).
    \end{equation}
    Since $\omega$ is a primitive element of $\F_q$, we have $p(x)r(x)s(x)=\omega^{3k^2}\prod_{i=0}^{n-1}(x-\omega^i)=x^{n}-1$. Let $p(x)=\sum_{i=1}^np_ix^i$, $r(x)=\sum_{i=1}^nr_ix^i$, $s(x)=\sum_{i=1}^ns_ix^i$.
    From $p(x)r(x)s(x)\equiv 0\mod (x^n-1)$ we have
    $$0=\sum_{d=0}^n\sum_{i+j+l\equiv d}p_ir_js_lx^d\quad\Longrightarrow\quad \sum_{i+j+l\equiv d}p_ir_js_l=0\mbox{ for all }d\le n-1.$$
    Therefore, modulo $(x^n-1,y^n-1,z^n-1)$ we have
    \begin{align*}
    a'(x,y,z)p(x)r(y)s(z)&=\sum_{i=0}^{n-1}\sum_{j=0}^{n-1}\sum_{l=0}^{n-1}x^iy^jz^l\sum_{i'=0}^n\sum_{j'=0}^n\sum_{l'=0}^n p_{i-i'}r_{j-j'}s_{l-l'}a'_{i'j'l'}\\
    &=\sum_{i=0}^{n-1}\sum_{j=0}^{n-1}\sum_{l=0}^{n-1}x^iy^jz^l\sum_{i'+j'+l'\equiv 0\mod n} p_{i-i'}r_{j-j'}s_{l-l'}\\
    &=\sum_{i=0}^{n-1}\sum_{j=0}^{n-1}\sum_{l=0}^{n-1}x^iy^jz^l\underbrace{\sum_{i''+j''+l''\equiv i+j+l\mod n} p_{i''}r_{j''}s_{l''}}_{=0}=0.
    \end{align*}
    (In the last line we used the substitutions $i'':=i-i'$, $j'':=j-j'$, $k'':=k-k'$).
    Thus, \eqref{eqn:a'-is-codeword} holds, hence \eqref{eqn:a-is-codeword} holds, therefore $a$ is a codeword of $C\boxplus C\boxplus C$ by Lemma \ref{lemma:sum-check}.

    By definition, $|a|=n^2$. Suppose $a=a_1+a_2+a_3$, where $a_1\in C\otimes \F_q^n\otimes \F_q^n$, $a_2\in \F_q^n\otimes C\otimes \F_q^n$, $a_3\in \F_q^n\otimes \F_q^n\otimes C$. Since each axis-parallel line in the cube $[n]^3$ covers only one non-zero element of $a_{ijl}$, we have $|a_1|_1+|a_2|_2+|a_3|_3\ge |a|=n^2$. 
    Taking into account $\|a\|=\frac{1}{n^3}|a|=\frac{1}{n}$, $\|a_i\|_i=\frac{1}{n^2}|a_i|_i$, we obtain 
    \[\sum_{i\in[3]}\|a_i\|_i=\frac{1}{n^2}\sum_{i\in[3]}|a_i|_i\ge 1 = n\|a\|.\]
    Therefore, $\rho(C,C,C)\le 1/n$.
\end{proof}
Lemma \ref{lemma:not-expanding} just proved shows that product expansion of the triple $(C,C,C)$ tends to zero as code length $n\to\infty$. 
Now let us combine known results to show that all tests $T_m^k$ are constantly robust for the code $C^{\otimes m}$ and $k\in[m-1]$ as $n\to\infty$.
First, we will show that the test $T_2^1$ is robust for code $C\otimes C$. Let us reformulate the theorem about robust testability of Reed-Solomon codes from \cite{Polishchuk:1994} for our case.
\begin{lemma}[{Corollary of \cite[Theorem 9]{Polishchuk:1994}
}]\label{th:RS-robust}
    Let $C$ be the $[n,k]$ primitive Reed-Solomon code over $\F_q$ defined by the check polynomial $(x-1)(x-\omega)\ldots (x-\omega^{k-1})$, where $n=q-1$, $k<n/2$, and $\omega$ is a primitive element of $\F_q$. Then for each $c_1\in C\otimes \F_q^n$, $c_2\in \F_q^n\otimes C$ if 
    $$\delta(c_1,c_2)\le \rbr{\frac 12-\frac kn}^2,$$ 
    then 
    $$\delta(c_1,C\otimes C)\le 2\delta(c_1,c_2),\qquad\delta(c_2,C\otimes C)\le 2\delta(c_1,c_2).$$
\end{lemma}
\begin{proof}
    Using discrete Fourier transform \cite[Theorem 6.1.5]{book:AlgCodes}, it is not hard to show that each codeword $c\in C$ can be defined as the vector of values of some polynomial $p_c\in \F_q[x]$ of degree at most $d=k-1$ at points $\rbr{1,\omega^{-1},\omega^{-2},\ldots,\omega^{1-n}}$. 
    We will use \cite[Theorem 9]{Polishchuk:1994} for $X=Y=\{1,\omega,...,\omega^{n-1}\}$, $d=k-1$, $\delta \in I$, where $I$ is the interval $\bigrbr{\sqrt{\delta(c_1,c_2)},\frac 12-\frac d n}$. 
    Since $\sqrt{\delta(c_1,c_2)}\le \frac 12-\frac kn<\frac{1}{2}-\frac{d}{n}$, the interval $I$ is not empty.
    
    Each codeword $c_1\in C\otimes \F_q^n$ (resp., $c_2\in \F_q^n\otimes C$) is defined by the vector of values on $X\times Y$ of some bivariate polynomial $p_{c_1}(x,y)$ of degree\footnote{We say that a polynomial $p(x,y)$ has degree $(a,b)$ if it has degree at most $a$ in $x$ and degree at most $b$ in $y$} $(d,n)$ (resp., $p_{c_2}(x,y)$ of degree $(n,d)$).
    For $c,c'\in \F_q^n\otimes \F_q^n$ we can interpret $\delta(c,c')$ as $\PP_{(x,y)\in X\times Y}\fbr{p_{c}(x,y)\ne p_{c'}(x,y)}$. 
    Since $\delta\in I$, the conditions of \cite[Theorem 9]{Polishchuk:1994} hold. Applying this theorem to $p_{c_1}$ and $p_{c_2}$ there exist $p_c$ of degree $(d,d)$ such that
    \begin{equation*}
      \PP_{(x,y)\in X\times Y}\fbr{p_{c_1}(x,y)\ne p_{c}(x,y)\mbox{ or }p_{c_2}(x,y)\ne p_{c}(x,y)}\le 2\delta^2  
    \end{equation*}
    The corresponding word $c$ belongs to the product code $C\otimes C$, since the degree of $p_c$ in each variable is bounded by $d$.
    Therefore, we have
    $$\delta(c_1,C\otimes C)\le \delta(c_1,c)=\PP_{(x,y)\in X\times Y}\fbr{p_{c_1}(x,y)\ne p_{c}(x,y)}\le 2\delta^2.$$
    Taking the infinum over all $\delta\in I$, we have $\delta(c_1,C\otimes C)\le 2\delta(c_1,c_2)$. Similarly, $\delta(c_2,C\otimes C)\le 2\delta(c_1,c_2)$.
\end{proof}

\begin{corollary}\label{col:RS-robust}
    $\rho_r(T_2^1, C\otimes C)\ge \frac{1}{72}$.
\end{corollary}
\begin{proof}
    Consider a word $x\in \F_q^{n\times n}$. Let $c_1$ and $c_2$ be the nearest words to $x$ from $C\otimes \F_q^n$ and $\F_q^n\otimes C$, respectively. Let $\alpha:=\delta(x,c_1)+\delta(x,c_2)$. We want to show that 
    \begin{equation*}
        \delta(x, C\otimes C)\le 36\rbr{\delta(x,C\otimes \F_q^n)+\delta(x,\F_q^n\otimes C)}.
    \end{equation*}
    By definition of $c_1$ and $c_2$ we have $\delta(x,c_1)=\delta(x,C\otimes \F_q^n)$, $\delta(x,c_2)=\delta(x,\F_q^n\otimes C)$, hence we need to prove that
    \begin{equation}\label{eqn:CxC-robust}
        \delta(x,C\otimes C)\le 36 \alpha.
    \end{equation}
    If $\alpha\ge \frac{1}{36}$, then \eqref{eqn:CxC-robust} holds. Now consider the main case $\alpha<\frac{1}{36}$. Since $C$ is $[n,k]$ code with $k=n/3$, in this case by the triangle inequality
    we have $\delta(c_1,c_2)\le\alpha<\frac{1}{36}=\rbr{\frac 12-\frac kn}^2$. 
    Hence, by Lemma \ref{th:RS-robust} we have 
    $$\delta(x,C\otimes C)\le \delta(x,c_1)+\delta(c_1,C\otimes C)\le \alpha+2\delta(c_1,c_2)\le 3\alpha.$$
    Thus, in this case \eqref{eqn:CxC-robust} holds as well, and the proof is complete.
\end{proof}

\newcommand{\cCn}[1]{\cC^{\widehat{(#1)}}}
\newcommand{\cCp}[1]{\cC^{(#1)}}
\newcommand{\prodC}[1][\cC]{\otimes #1}

\begin{lemma}[Robustness of test composition]\label{lemma:robust-combine}
   Let $\cC\subseteq \F_q^n$ and $1\le k_1<k_2<m$. Then
   \[\rho_r(T_m^{k_1},\cC^{\otimes m})\ge \rho_r(T_m^{k_2},\cC^{\otimes m})\rho_r(T_{k_2}^{k_1},\cC^{\otimes k_2}).\]
\end{lemma}
\begin{proof}
Fix $x\in (\F_q^n)^{\otimes m}$. For each $k\in[m-1]$ and $\pi\in T_m^{k}$ we have $\cCpow{m}|_\pi\cong \cCpow{k}$, hence
\[
\EE_{\pi\in T_m^{k}}\delta(x|_{\pi}, \cCpow{m}|_\pi)=\EE_{\pi\in T_m^{k}}\delta(x|_{\pi}, \cCpow{k_1}).
\]
Therefore, 
\begin{align*}
\delta(x, \cCpow{m})\rho_r(T_m^{k_2},\cCpow{m})\rho_r(T_{k_2}^{k_1},\cCpow{k_2})
\le &\EE_{\pi\in T_m^{k_2}}\delta(x|_\pi,\cCpow{k_2})\rho_r(T_{k_2}^{k_1},\cCpow{k_2})\\
\le &\EE_{\pi\in T_m^{k_2}}\EE_{\pi'\in T_{k_2}^{k_1}(\pi)}\delta(x|_{\pi'},\cCpow{k_1}) = \EE_{\pi'\in T_m^{k_1}}\delta(x|_{\pi'},\cCpow{k_1}).\qedhere
\end{align*}
\end{proof}

\begin{lemma}\label{lemma:robust-Tm1}
    Let $\cC\subseteq \F_q^n$. Denote $M:=\frac12(m-2)(m+3)=\sum_{k=3}^m k$. Then \[
    \rho_r(T_m^1,\cCpow{m})\ge \frac1{12^{m-2}}\cdot\rho_r(T_2^1,\cCpow{2})\cdot\delta(\cC)^{M}.
    \]
\end{lemma}
\begin{proof}
    From \cite[Theorem 2.6]{Chiesa2020:hyperplane-robust} we have a lower bound on robustness of axis-parallel hyperplane test:
    \begin{equation}\label{eqn:hp-robust}
    \rho_r(T_k^{k-1}, \cCpow{k})\ge \frac{1}{12}\delta(\cC)^k.
    \end{equation}
    Applying Lemma \ref{lemma:robust-combine} repeatedly $m-2$ times, then using the inequality \eqref{eqn:hp-robust}, we obtain:
    \[
    \rho_r(T_m^1,\cCpow{m})\ge \rho_r(T_2^1,\cCpow{2})\prod_{k=3}^m\rho_r(T_k^{k-1},\cCpow{k})\ge \frac{1}{12^{m-2}}\cdot\rho_r(T_2^1,\cCpow{2})\cdot\delta(\cC)^{M}.\qedhere
    \]
\end{proof}
Now we are ready to prove Theorem \ref{th:main} and Proposition \ref{prop:main}.
\MainTheorem*
\begin{proof}
    Claim 1 of the theorem follows from Lemma \ref{lemma:not-expanding} and \cite[Lemma 11]{PK2022robust}:
    \[
    \rho(\underbrace{C,\ldots,C}_{m\ge 3\ \mbox{\scriptsize\rm times}})\le \rho(C,C,C)\le 1/n.
    \]
    Claim 2 of the theorem follows from Lemma \ref{lemma:robust-Tm1} 
    and Corollary \ref{col:RS-robust}. Recall that $C$ is $[n,\frac n3,\frac23 n+1]$ Reed-Solomon code, therefore $\delta(C)=\frac23+\frac1n$. Put $\alpha_r:=\frac{1}{72\cdot 12^{m-2}}\cdot \rbr{\frac23}^{\frac12(m-2)(m+3)}$. By Lemma \ref{lemma:robust-Tm1} and Corollary \ref{col:RS-robust} we have
    \[
    \rho_r(T_m^1,C^{\otimes m})
    \ge \rho_r(T_2^1, C^{\otimes 2})\cdot\frac{1}{12^{m-2}}\cdot\delta(C)^{\frac12(m-2)(m+3)} 
    >\frac{1}{72}\cdot\frac{1}{12^{m-2}}\cdot\rbr{\frac23}^{\frac12(m-2)(m+3)}=\alpha_r.
    \]
    Claim 3 of the theorem with $\alpha_a:=\frac23\frac{\alpha_r}{1+\alpha_r}$ follows from Claim 2 and Lemma \ref{lemma:robust=agreement}.
\end{proof}

\MainProp*
\begin{proof}
    Let $\rho:=\rho(\underbrace{\cC,\ldots,\cC}_{m\ \mbox{\scriptsize\rm times}})$. The proof is the sequence of following steps.
    \begin{itemize}
        \item By \cite[Lemma 11]{PK2022robust} we have $\rho(\cC,\cC)\ge\rho$, $\delta(\cC)=\rho(\cC)\ge\rho$.
        \item Using \cite[Lemma 1]{PK2022robust}, we obtain $\rho_a(\cCpow{2})\ge\rho(\cC,\cC)\ge\rho$.
        \item From \cite[Lemma 2.9]{Dinur:stoc2022} we have 
        \begin{equation*}
        \rho_r(T_2^1, \cCpow{2})\ge \frac{\rho_a(\cCpow{2})}{2(1+\rho_a(\cCpow{2}))}\ge \frac{\rho}{2(\rho+1)}\ge\frac14\rho.    
        \end{equation*}
        \item Finally, by Lemma \ref{lemma:robust-Tm1} we have \[
        \rho_r(T_m^1,\cCpow{m})\ge\rbr{\frac1{12}}^{m-2}\rho_r(T_2^1,\cCpow{2})\cdot\delta(\cC)^{\frac12(m-2)(m+3)}\ge \frac{1}{12^{m-2}}\cdot\frac{1}{4}\rho\cdot \rho^{\frac12(m-2)(m+3)}.
        \]
    \end{itemize}
    Thus, we obtain the required inequality with $\alpha(\rho)=\frac{1}{4\cdot 12^{m-2}}\rho^{\frac12(m-2)(m+3)+1}$.
\end{proof}

\section*{Acknowledgment}

This work was supported by the Ministry of Science and Higher Education of the Russian Federation (Grant 075-15-2020-801).

\bibliographystyle{IEEEtran}
\bibliography{codes.bib}

\appendix

\section{Relation between robust and agreement testability}\label{app:robust=agreement}
In this section we prove Lemma \ref{lemma:robust=agreement}, which states that robust and agreement testability are the same up to a constant factor for the axis-parallel line test for product codes.

\LemRobustAgreement*
\begin{proof}
    \textbf{1. Agreement testability implies robust testability.} Consider arbitrary $x\in \F_q^{n_1\times \cdots\times n_m}$. Let $y_i:=\argmin_{y\in \cC^{(i)}}\|y-x\|$ for all $i\in [m]$. There exists $z\in\otimes\cC$ such that 
    \[
    \rho_a\EE_{i\in[m]} \|y_i-z\|_i\le \EE_{i,j\in[m]}\|y_i-y_j\|.
    \]
    Denote $d_x:=\EE_{\ell\in T_m^1}\delta(x|_\ell, \otimes\cC|_{\ell})$. We have 
    \[
    d_x = \EE_{i\in[m]}\EE_{\ell\in\cL_i}\delta(x|_\ell, \cC_i)=\EE_{i\in[m]}\delta(x,\cC^{(i)})=\EE_{i\in[m]}\|x-y_i\|.
    \]
    Hence
    \begin{align*}
    \|x-z\| \le \EE_{i\in[m]}(\|x-y_i\|+\underbrace{\|y_i-z\|}_{\le\|y_i-z\|_i}) 
    &\le d_x+\frac1{\rho_a}\EE_{i,j\in[m]}\|y_i-y_j\|\\
    &\le d_x+\frac1{\rho_a}\cdot 2\EE_{i\in [m]}\|x-y_i\|=d_x\rbr{1+\frac{2}{\rho_a}}\le \frac{4}{\rho_a}d_x.
    \end{align*}
    Therefore, $\rho_r\ge \rho_a/4$.
    
    \textbf{2. Robust testability implies agreement testability.} 
    Consider arbitrary words $c_i\in \cC^{(i)}$ for $i\in [m]$.
    Let 
    $$i_0:=\argmin_{i\in[m]}\EE_{j\in [m]}\delta(c_i, c_j).$$
    Thus we have
    \begin{equation}\label{eqn:cI0-J}
    \EE_{j\in[m]}\|c_{i_0}- c_j\|\le \EE_{i,j\in[m]}\|c_i - c_j\|.
    \end{equation}
    Since the test $T_m^1$ is $\rho_r$-robust for $\otimes\cC$, there exists $c\in\otimes\cC$ such that 
    \begin{equation}\label{eqn:cI0-c}
    \rho_r \|c_{i_0}-c\|\le \EE_{\ell\in T_{m}^1}\delta(c_{i_0}|_{\ell}, \cC|_{\ell})= \EE_{j\in [m]}\delta(c_{i_0}, \cC^{(j)})\le \EE_{j\in[m]}\|c_{i_0} - c_j\|\le \EE_{i,j\in[m]}\|c_i- c_j\|.
    \end{equation}
    Let $\delta_*:=\min_{i\in[m]}\delta(\cC_i)$. For $i\in[m]$, $x\in\cC^{(i)}$ we have $\delta(\cC_i)\|x\|_i\le \|x\|$. Hence, applying the triangle inequality, followed by \eqref{eqn:cI0-J} and \eqref{eqn:cI0-c}, we get:
    \[
    \delta_*\EE_{i\in[m]}\|c_i-c\|_i\le \EE_{i\in[m]}\|c_i-c\|\le \|c_{i_0}-c\|+\EE_{i\in[m]}\|c_i-c_{i_0}\|
    \le \rbr{1+\frac1{\rho_r}}\EE_{i,j\in[m]}\|c_i-c_j\|.
    \]
    Therefore, 
    $\rho_a(T_m^k, \cC^{\otimes m})\ge\delta_*(1+\frac1{\rho_r})^{-1}$.
\end{proof}
\end{document}